\documentclass[reqno,12pt]{amsart}

\usepackage[T1]{fontenc}
\usepackage[utf8]{inputenc}
\usepackage[top=0.9in,left=0.95in,right=0.95in,bottom=0.9in]{geometry}
\usepackage{amsmath,amssymb,amsthm,graphicx,mathrsfs}
\usepackage{bm}
\usepackage{braket}
\usepackage{dsfont}
\usepackage{mathbbol}
\usepackage{extarrows}
\usepackage{subcaption}
\usepackage{enumitem}
\usepackage[dvipsnames]{xcolor}
\usepackage{tikz}
\usepackage[unicode=true,pdfusetitle,bookmarks=true,bookmarksnumbered=false,bookmarksopen=false,breaklinks=false,pdfborder={0 0 1},backref=false,colorlinks=true,linkcolor=violet,citecolor=YellowOrange!85!black,urlcolor=Aquamarine!85!black]{hyperref}
\usepackage[nameinlink]{cleveref}

\numberwithin{equation}{section}
\DeclareSymbolFontAlphabet{\amsmathbb}{AMSb}

\newcommand{\safeincludegraphics}[2][]{%
  \IfFileExists{#2}{\includegraphics[#1]{#2}}{%
    \fbox{\parbox{0.82\linewidth}{\centering Missing figure file:\\[2pt]\texttt{\detokenize{#2}}}}%
  }%
}

\newcommand{\Range}{{\operatorname{Range}}}

\newcommand{\R}{\amsmathbb{R}}

\newcommand{\N}{\amsmathbb{N}}
\newcommand{\Hh}{\amsmathbb{H}}
\newcommand{\1}{\mathds{1}}

\newcommand{\Bb}{\amsmathbb{B}}
\newcommand{\Z}{\amsmathbb{Z}}

\newcommand{\C}{\amsmathbb{C}}

\newcommand{\p}{{\partial}}

\newcommand{\Span}{{\operatorname{Span}}}
\newcommand{\Tr}{{\operatorname{Tr}}}

\newcommand{\lr}[1]{\langle #1 \rangle}

\newcommand{\XX}{\mathcal{X}}
\newcommand{\GG}{\mathcal{G}}
\newcommand{\EE}{\mathcal{E}}
\newcommand{\II}{\mathcal{I}}
\newcommand{\UU}{{\mathcal{U}}}
\newcommand{\PP}{{\mathcal{P}}}

\newcommand{\QQ}{{\mathcal{Q}}}

\newcommand{\AAA}{{\mathcal{A}}}
\newcommand{\NN}{\mathcal{N}}
\newcommand{\supp}{\mathrm{supp}}
\newcommand{\BBB}{\mathcal{B}}
\newcommand{\HS}{{\operatorname{HS}}}

\newcommand{\systeme}[1]{\left\{ \begin{matrix} #1 \end{matrix} \right.}

\newcommand{\de}{ \ \mathrel{\stackrel{\makebox[0pt]{\mbox{\normalfont\tiny def}}}{=}} \ }
\newcommand{\Spec}{\operatorname{Spec}}
\newcommand{\Int}{\operatorname{int}}

\newcommand{\MG}{{\mathcal G}}

\newcommand{\br}[1]{\left(#1\right)}
\newcommand{\Id}{\mathds{1}}

\newcommand{\bbL}{\mathbb{\Lambda}}

\crefname{section}{\S}{\S}
\crefname{equation}{}{equations}
\Crefname{equation}{Equation}{Equations}
\crefrangelabelformat{equation}{(#3#1#4--#5#2#6)}
\crefmultiformat{equation}{(#2#1#3}{, #2#1#3)}{#2#1#3}{#2#1#3}
\Crefmultiformat{equation}{(#2#1#3}{, #2#1#3)}{#2#1#3}{#2#1#3}
\title[Geometric bulk-edge correspondence for Z2-topological insulators]{Geometric bulk-edge correspondence for $\Z_2$-topological insulators}
\author{Alexis Drouot}
\address[Alexis Drouot]{University of Washington, Seattle, USA.}
\email{adrouot@uw.edu}
\author{Jacob Shapiro}
\address[Jacob Shapiro]{Princeton University, Princeton, USA.}
\email{shapiro@math.princeton.edu}
\author{Xiaowen Zhu}
\address[Xiaowen Zhu]{University of Minnesota, Twin Cities, Minnesapolis, USA}
\email{xwzhu@umn.edu}

\theoremstyle{plain}
\newtheorem{thm}{Theorem}
\newtheorem{lemma}{Lemma}[section]

\newtheorem{corollary}[lemma]{Corollary}
\newtheorem{proposition}{Proposition}
\newtheorem{theorem}[thm]{Theorem}

\theoremstyle{definition}
\newtheorem{definition}{Definition}

\newtheorem{remark}{Remark}[section]

\begin{document}
\begin{abstract}
Fermionic time-reversal-invariant insulators in two dimensions---class AII in the Kitaev table---come in two topological phases. These phases are characterized by a $\Z_2$-valued invariant, the Fu--Kane--Mele index. We prove a geometric bulk-edge correspondence for curved interfaces: if two such insulators occupy complementary regions separated by a curved boundary, then the $\Z_2$ edge index of the interface system is the product, modulo two, of the difference of the two bulk $\Z_2$ indices and a geometric intersection number associated with the boundary and the measurement region. The argument is a $\Z_2$ analogue of the curved-interface connection formula proved for Hall insulators in \cite{DZ24}.
\end{abstract}
\maketitle

\section{Introduction}

Topological insulators are remarkable phases of matter that block conduction in their bulk but support robust currents along their boundaries. They are characterized by a (non-zero) topological marker that, strikingly, is related to the conductance of the edge. This principle is called the bulk-edge correspondence.

The original examples of topological phases of matter are Hall insulators, in connection to the famous quantum Hall effect. Their Hall conductance is a quantized topological quantity, and the bulk-edge correspondence asserts that it is equal to the conductance of boundaries. In particular, one-way currents emerge when this index is non-zero. This can only happen when time-reversal invariance is broken, for instance by a magnetic field.

The search for novel topological phases of matter has since shifted to systems with time-reversal invariance. Building on their work on quantum spin phases \cite{KM05,FKM07,P09}, Fu--Kane--Mele discovered $\Z_2$-insulators, a class of materials with fermionic time-reversal invariance and a topological index with values modulo $2$.

Whether for Hall or $\Z_2$-insulators, the bulk-edge correspondence has mostly been analyzed in the context of straight interfaces \cite{EG02,EGS05,GS16,FSSWY20,BC24}. More recently, this has changed through a series of works \cite{LT22,DZ23,DZ24,DSZ24} which studied this principle for Hall insulators and curved edges. It culminated with a formula relating a geometry-dependent intersection number, the Hall conductance, and the edge conductance in \cite{DZ24}. The objective of this work is to carry out an analogous program for $\Z_2$-insulators. We prove that the curved $\Z_2$ edge index
    \begin{align}\label{eq-1a}
        \NN(H_e) = \XX_{\Omega,W} \cdot \big( \II(H_+) - \II(H_-)\big) \mod 2.
    \end{align}
In \eqref{eq-1a}, $H_+$ and $H_-$ are the two bulk Hamiltonians, $H_e$ is an interface Hamiltonian modeled on $H_+$ in $\Omega$ and on $H_-$ in $\Omega^c$, $\NN(H_e)$ is the $\Z_2$ edge index measured in a set $W$ transverse to $\Omega$, and $\XX_{\Omega,W}$ is the corresponding geometric intersection number.

We organize the paper as follows:
\begin{itemize}
    \item We review the general theory of $\Z_2$-indices, originating in \cite{Atiyah1969}, including the definitions and algebraic and analytic properties, see \S \ref{sec-2};
    \item We list and prove key properties of the curved $\Z_2$ indices $\NN(H_e)$ and $\II(H_\pm)$, including additivity in $\Omega$ and $W$ and invariance under compact deformation, see \S \ref{sec-3};
    \item We reduce the edge index to geometric $\Z_2$-indices using Fredholm theory, see \Cref{thm-2} and \S \ref{sec-4};
    \item We adapt the geometric framework developed in \cite{DZ24} (for Hall insulators) to deal with curved interfaces, see \Cref{thm-3} and \S \ref{sec-5}.
\end{itemize}

\subsection{Time-reversal-invariant insulators}
Let $H$ be a bounded self-adjoint Hamiltonian on $\ell^2(\Z^2,\C^d)$.
\begin{definition}\label{def:time reversal map} A \textit{(Fermionic) time-reversal map} $\Theta$ is an antiunitary operator\footnote{An operator $U$ on a Hilbert space $\mathcal H$ is antiunitary if it is a bijective antilinear map such that $\lr{Ux, Uy} = \lr{y, x} = \overline{\lr{x, y}}$ for all $x, y \in \mathcal H$.} on $\ell^2(\Z^2,\C^d)$ such that $\Theta^2 = -\Id$ and $[\Theta,x_j] = 0$ for $j=1,2$, where $x_j$ is multiplication by the $j$th coordinate on $\Z^2$.

We say a Hamiltonian $H$ is \textit{$\Theta$-invariant} if $[H, \Theta] = 0$. We say $H$ is \textit{time-reversal invariant} if it is $\Theta$-invariant for some time-reversal operator $\Theta$.
\end{definition}
\begin{definition}  We say that $H$~is:
\begin{itemize}
    \item \textit{exponentially short-range (ESR)} if its kernel $H(x,y)$ satisfies, for some $\nu > 0$,
\begin{align}
    | H(x,y)| \leq \nu^{-1} e^{-2\nu|x-y|}, \qquad x,y \in \Z^2;
\end{align}
    \item \textit{polynomially short-range (PSR)} if its kernel $H(x,y)$ satisfies, for every $N>0$, there is $C_N$ such that
    \begin{align}
        | H(x,y)| \leq C_N |x - y|^{-N}, \qquad x, y \in \Z^2;
    \end{align}
    \end{itemize}
\end{definition}

Time-reversal-invariant insulators come in two topological phases, characterized by a $\Z_2$-valued index. It was first constructed by Fu--Kane--Mele \cite{KM05,FKM07,Fu_Kane_2007} for translation-invariant Hamiltonians. Their definition was later extended by Schulz--Baldes \cite{SB_2015} without spatial symmetry assumptions. Here, we use a definition of the index due to Kitaev \cite{Kitaev2009}. This is the $\Z_2$-analogue of the Hall conductance/bulk index used in \cite{DZ23, DZ24}.

\begin{definition}  The $\Z_2$ index of a local, $\Theta$-invariant Hamiltonian $H$ with spectral gap $\GG$ is
\begin{equation}\label{eqdef-bulk}
\II(H) = \dim \ker \big(\1_{\Hh_1} \cdot e^{-2\pi i \1_{\Hh_2} \Pi \1_{\Hh_2}} \cdot \1_{\Hh_1}+\1_{\Hh_1^c}\big) \mod 2, \qquad \text{where}
\end{equation}
\begin{itemize}
    \item $\Pi = \1_{(-\infty,\inf \GG]}(H)$ is the spectral projection below the spectral gap $\GG$;
    \item $\Hh_1,\Hh_2$ are, respectively, the right and upper half-planes.
\end{itemize}
\end{definition}
\begin{figure}
    \centering
    \safeincludegraphics[width=0.7\linewidth]{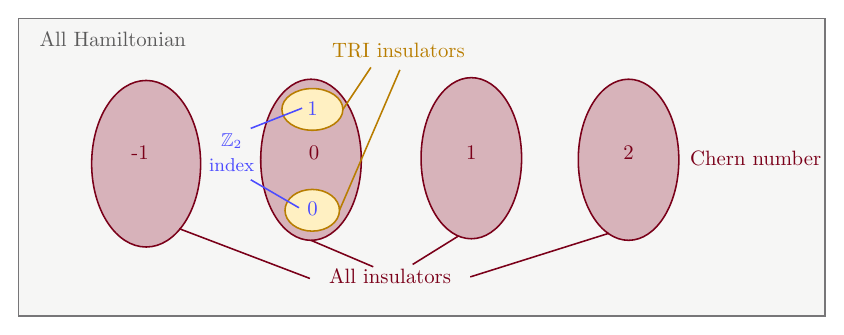}
    \caption{The gray box denotes short-range self-adjoint Hamiltonians. The maroon ellipses denote insulating Hamiltonians, which can be classified by Chern numbers. The gold ellipses denote time-reversal-invariant insulators. These have vanishing Chern number but may have different $\Z_2$ indices, indicated in blue.}
    \label{fig:classification}
\end{figure}
In the rest of this paper, we use the notation ``$\equiv$'' for equality modulo $2$, and
\[
U_\Omega := e^{-2\pi i \1_\Omega \Pi \1_\Omega}, \qquad \bbL_W U_\Omega := \1_{W^c} + \1_W U_\Omega \1_W
\]
for $\Omega, W\subset \Z^2$. Under these notations,
\begin{align}
\II(H) \equiv \dim \ker (\bbL_{\Hh_1}U_{\Hh_2}).
\end{align}
This definition has a clearer structure in the general framework of \S \ref{sec-2}: the $\Z_2$ index is defined for a $\Theta$-admissible pair consisting of a projection $\PP=\1_W$ and a unitary operator $\UU=U_\Omega$. In \S \ref{sec-3}, we apply this framework to time-reversal-invariant Hamiltonians.

\subsection{Interface systems} Suppose we have two time-reversal invariant insulators $H_\pm$ with distinct $\Z_2$-index, one occupying $\Omega\subset \R^2$, one occupying $\Omega^c$. We consider the resulting interface Hamiltonian which connects these two insulators.

\begin{definition}\label{def:DW edge Hamiltonian} Let $\Omega \subset \R^2$ and $H_\pm$ be two ESR Hamiltonians that are $\Theta$-invariant and insulating in $\GG$. An \textit{interface Hamiltonian $H_e$ that is compatible with $H_\pm$ and $\Omega$} is an ESR self-adjoint bounded operator $H_e$ such that for some $\nu>0$,
\begin{equation}
    \forall x,y\in\Z^2, \qquad |E(x,y)|\leq\nu^{-1} e^{-\nu d\left(x,\partial \Omega\right)},  \qquad E \de H_e- \1_\Omega H_+ \1_\Omega - \1_{\Omega^c} H_- \1_{\Omega^c}.
\end{equation}
We call $H_\pm$ the \textit{bulk operators}, $\MG$ the \textit{bulk (spectral) gap}, and $\II(H_\pm)$ the \textit{bulk index}.
\end{definition}

To define the edge index for the interface system, we need to introduce a geometric condition, called transversality, for two sets $\Omega$, $W\subset \R^2$, as in \cite{DZ24}. 
\begin{definition}\label{def:trans sets}
    We say that two sets $\Omega,W\subset \R^2$ are transverse (see \Cref{fig: transversality}) if
\begin{equation}\label{eq:transversal sets}
    \liminf_{\|x\| \rightarrow +\infty} \dfrac{\log \Psi_{\Omega,W}(x)}{\log \|x\|} >0, \qquad \Psi_{\Omega,W}(x) = 1+d(x,\p \Omega) + d(x,\p W).
\end{equation}
\end{definition}

\begin{figure}[t]
     \centering
     \begin{subfigure}[ht]{0.3\textwidth}
         \centering
         \includegraphics[width=0.9\textwidth]{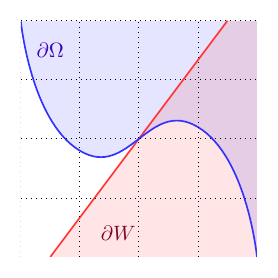}
         \caption{}
         \label{subfig: T_1}
     \end{subfigure}
     \hfill
     \begin{subfigure}[ht]{0.3\textwidth}
         \centering
         \includegraphics[width=0.9\textwidth]{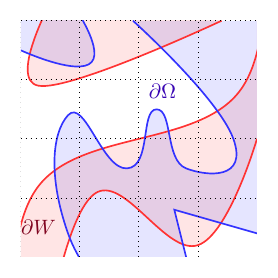}
         \caption{}
         \label{subfig: T_2}
     \end{subfigure}
     \hfill
     \begin{subfigure}[ht]{0.3\textwidth}
         \centering
         \includegraphics[width= 0.9\textwidth]{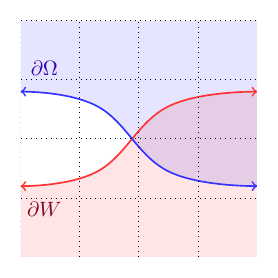}
         \caption{}
     \end{subfigure}
        \caption{In subfigures $(a)$ and $(b)$, the pair $(\Omega,W)$ is transversal, since $\partial \Omega$ and $\partial W$ eventually separate at least polynomially fast. In contrast, in subfigure $(c)$, $\partial \Omega$ and $\partial W$ become asymptotically parallel. Hence $(\Omega,W)$ is not transversal in the sense of \Cref{def:trans sets}. Subfigure $(b)$ highlights the generality allowed by the definition.}
        \label{fig: transversality}
\end{figure}

For two transverse sets $\Omega, W$, one can define an intersection number $\XX_{\Omega, W}$ between $\p \Omega$ and $\p W$, which essentially counts the number of intersections between $\p \Omega$ and $\p W$, see \Cref{fig: 1} for an illustration and \Cref{sec-5} for a short review in this context and \cite[Sec. 7]{DZ24} for all details. We can now define the edge index.

\begin{figure}[b]
     \centering
     \begin{subfigure}[ht]{0.3\textwidth}
         \centering
         \safeincludegraphics[width=0.9\textwidth]{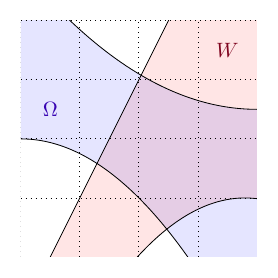}
         \caption{}
         \label{fig: 1(a)}
     \end{subfigure}
     \hfill
     \begin{subfigure}[ht]{0.3\textwidth}
         \centering
         \safeincludegraphics[width=0.9\textwidth]{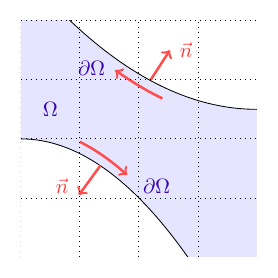}
         \caption{}
         \label{fig: 1(b)}
     \end{subfigure}
     \hfill
     \begin{subfigure}[ht]{0.3\textwidth}
         \centering
         \safeincludegraphics[width=0.9\textwidth]{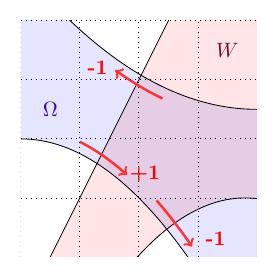}
         \caption{}
         \label{fig: 1(c)}
     \end{subfigure}
        \caption{We define the intersection number $\mathcal X_{\Omega, W}$ between transverse simple sets $\Omega, W$ in two steps. We first orient $\p \Omega$ such that $\Omega$ is to its left according to the outward-pointing normal, see (a) and (b); then we count how many times the oriented $\p \Omega$ enters $W$, see (c). Here $\mathcal X_{\Omega, W} = +1 - 1 - 1 = -1$.}
        \label{fig: 1}
\end{figure}

\begin{figure}[t]
     \centering
     \begin{subfigure}[b]{0.3\textwidth}
         \centering
         \safeincludegraphics[width=0.9\textwidth]{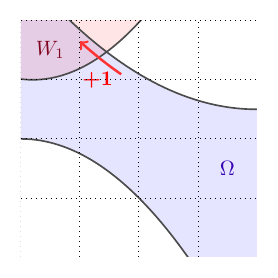}
         \caption{}
         \label{fig: 2(a)}
     \end{subfigure}
     \hfill
     \begin{subfigure}[b]{0.3\textwidth}
         \centering
         \safeincludegraphics[width=0.9\textwidth]{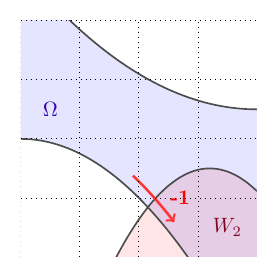}
         \caption{}
         \label{fig: 2(b)}
     \end{subfigure}
     \hfill
     \begin{subfigure}[b]{0.3\textwidth}
         \centering
         \safeincludegraphics[width=0.9\textwidth]{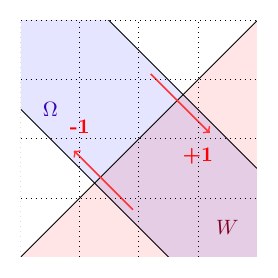}
         \caption{}
         \label{fig: 2(c)}
     \end{subfigure}
        \caption{The intersection numbers in subfigures \ref{fig: 2(a)}, \ref{fig: 2(b)}, and \ref{fig: 2(c)} are $+1$, $-1$, and $0$, respectively. In the integer-valued Hall setting, the corresponding edge conductance changes sign or vanishes, respectively. }
        \label{fig: 2}
\end{figure}

\begin{definition}\label{def-edge_index} Let $H_e$ be an interface Hamiltonian compatible with $H_\pm$ and $\Omega$, with bulk spectral gap $\GG$. The edge index of $H_e$ in $\GG$ is
\begin{equation}\label{eqdef-edge}
    \NN(H_e) \equiv  \dim \ker \br{\bbL_W e^{-2\pi i g(H_e)}}
\end{equation}
where $g \in C^\infty(\R)$ is $1$ on $(-\infty,\inf \GG]$ and $0$ on $[\sup \GG, +\infty)$.
\end{definition}

\subsection{Main result} The main result of this work is a geometric bulk-edge correspondence:

\begin{theorem}\label{thm-BEC}
    Assume that $H_\pm$ are two ESR, $\Theta$-invariant Hamiltonians with a joint spectral gap $\GG \subset \R$, $\Omega$ and $W$ are two transverse sets, and $H_e$ is an interface edge Hamiltonian compatible with $H_\pm$ and $\Omega$. Then
    \begin{align}
        \NN(H_e) = \XX_{\Omega,W} \cdot \big( \II(H_+) - \II(H_-)\big) \mod 2.
    \end{align}
\end{theorem}

\Cref{thm-BEC} connects the bulk index to the edge index through the geometric intersection number $\XX_{\Omega,W}$, which quantifies the relative position of $\Omega$ and $W$. The same identity was previously derived for $\Z$-valued Hall insulators in \cite{DZ24}.

One direct application of \Cref{thm-BEC} is the characterization of the shape of $\Omega$ that admits absolutely continuous edge spectrum filling the bulk spectral gap -- this is an extension of previous \cite{DSZ24} from Hall insulators to $\Z_2$-topological insulators. 
\begin{definition}[see also \cite{DSZ24}]
    We say $\Omega$ and $\Omega^c$ both contain a parabolic region if there exists $\alpha > 0$ and two rigid motions $M_\pm$ of $\R^2$  such that 
\begin{eqnarray}\label{eq-1j}
    \{ M_+x : \ x_2 \geq |x_1|^\alpha \} \subset \Omega, \qquad \{ M_-x : \ x_2 \geq |x_1|^\alpha \} \subset \Omega^c.
\end{eqnarray}
\end{definition}
And we have 
\begin{corollary}
    Assume that:
    \begin{itemize}
        \item $H_\pm$ are two ESR, $\Theta$-invariant Hamiltonians with a joint spectral gap $\GG \subset \R$ and $\II(H_+) \neq \II(H_-) \mod 2$;
        \item both $\Omega$ and $\Omega^c$ contain a parabolic region;
        \item $H_e$ is an interface edge Hamiltonian compatible with $H_\pm$ and $\Omega$.
    \end{itemize}
 Then $\GG$ is contained in the absolutely continuous spectrum of $H_e$. 
\end{corollary}
For Hall insulators, this result follows from a combination of geometric bulk-edge correspondence \cite{DZ24}, absolutely continuous spectrum for a pair of unitary and projection with non-zero index by Bols--Werner \cite{BW22} (first proved by Arsh--Bourget--Joye \cite{ABJ20}), and some spectral theory arguments, see \cite{DSZ24} or \cite{BW22}. 

For $\Z_2$-topological insulator, we just need to update the Chern geometric bulk-edge correspondence with \Cref{thm-BEC} and replace \cite{BW22} with its $\Z_2$-version Bols--Cedzich \cite{BC24} and the corollary would follow immediately as in \cite{DSZ24}.  

\subsection{Strategy of proofs of main results}

The proof of \Cref{thm-BEC} follows the strategy of \cite{DZ24}. We first show that the edge index $\NN(H_e)$ is the difference of two geometric bulk indices:
\begin{theorem}\label{thm-2} Under the assumptions of Theorem \ref{thm-BEC},
\begin{equation}
    \NN(H_e) \equiv \II_{\Omega,W}(H_+) - \II_{\Omega,W}(H_-).
\end{equation}
\end{theorem}
Here, the geometric bulk index $\II_{\Omega, W}(H)$ is defined as the curved analogue of $\II(H_\pm)$ in \eqref{eqdef-bulk}, obtained by replacing the half-planes $\Hh_1$ and $\Hh_2$ with transverse sets $\Omega$ and $W$:
\begin{equation}\label{eqdef-geobulk}
    \II_{\Omega,W}(H) \equiv \dim \ker \big(\1_W \cdot e^{-2i\pi \1_\Omega \Pi \1_\Omega} \cdot \1_W+\1_{W^c} \big) \equiv \dim\ker(\bbL_W U_\Omega).
\end{equation}
Transversality guarantees that \eqref{eqdef-geobulk} is well-defined; see \Cref{sec-3}.

In \cite{DZ24}, the analogous statement for Hall insulators was proved using a spectral-theoretic approach originating in \cite{EGS05}. Here we use a direct index-theoretic argument adapted to $\Z_2$-insulators. In outline, we show
\begin{align}
    \NN(H_e) &= \II(e^{-2\pi i g(H_e)}, \1_W) && \text{\textit{by definition}}\\
    &= \II(e^{-2\pi i (\1_\Omega g(H_+) \1_\Omega + \1_{\Omega^c} g(H_-) \1_{\Omega^c})}, \1_W) && \text{\textit{by stability}}\\
    &= \II(e^{-2\pi i (\1_\Omega g(H_+) \1_\Omega)}, \1_W) + \II(e^{-2\pi i (\1_{\Omega^c} g(H_-) \1_{\Omega^c})}, \1_W) && \text{\textit{by additivity}}\\
    &\equiv \II_{\Omega, W}(H_+) - \II_{\Omega, W}(H_-) && \text{\textit{by complement}}.
\end{align}
The stability step replaces $g(H_e)$ by
\begin{align}
    \1_\Omega g(H_+) \1_\Omega + \1_{\Omega^c} g(H_-) \1_{\Omega^c} = \1_\Omega \Pi_+ \1_\Omega + \1_{\Omega^c} \Pi_- \1_{\Omega^c}.
\end{align}
This replacement does not change the $\Z_2$ index by the stability results proved in \S\ref{sec-3}. That section also proves the additivity and complement identities used in the last two lines.

We then show that one can extract the intersection number $\XX_{\Omega,W}$ out of the geometric bulk indices $\II_{\Omega,W}(H_\pm)$ which is similar to \cite{DZ24} but we believe is the first result of this type for $\Z_2$-index:

\begin{theorem}\label{thm-3} Under the assumptions of Theorem \ref{thm-BEC},
\begin{equation}
    \II_{\Omega,W}(H_\pm) \equiv \XX_{\Omega,W} \cdot \II(H_\pm).
\end{equation}
\end{theorem}

This result relies on the algebraic properties of $\II_{\Omega,W}$ and on the geometric framework developed in \cite{DZ24}. The proof has two parts:
\begin{enumerate}
    \item By additivity of $\II_{\Omega,W}$ in both $\Omega$ and $W$, we can reduce the proof of Theorem \ref{thm-3} to the case where both $\p \Omega$ and $\p W$ are made of a single curve (so $\Omega$ and $W$ are simple sets), see \Cref{subsec-5.2}.
    \item In the simple-set case, we use local or global uniform deformations to deform $\Omega$ and $W$, while preserving the index, to shapes whose index is easy to compute, e.g. non-intersection or half-plane-intersection. In particular, the hardest case is when $\p \Omega$ and $\p W$ intersect once. In this case, uniformity in global deformation allows us to pass the computation of the index to the limit as in \cite{DZ24}, which leads to the conclusion of \Cref{thm-3}, see \Cref{subsec-5.1}.
\end{enumerate}

\subsection{Structure of the paper}
The rest of the paper is organized as follows.
\begin{itemize}
    \item  In \Cref{sec-2}, we discuss the general theory of $\Z_2$-index for a pair of $\Theta$-admissible unitary and projection.
    \item In \Cref{sec-3}, we derive key properties of the $\Z_2$-index for time-reversal-invariant Hamiltonians.
    \item In \Cref{sec-4}, we prove \Cref{thm-2}.
    \item In \Cref{sec-5}, we prove \Cref{thm-3}.
\end{itemize}

\subsection{Notations}

\begin{itemize}
    \item We say $H$ is antilinear if $H(\lambda x) = \overline{\lambda} H(x)$.
    \item The adjoint $H^*$ of an antilinear operator $H$ is given by $\lr{x, H^*y} = \overline{\lr{Hx, y}}$. It is also antilinear.
    \item We say an operator $H$ is antiunitary if $\lr{Hx, Hy} = \lr{y, x}$.
    \item We say $A$ is $\Theta$-invariant if $\Theta A \Theta^{-1} = A$.
    \item We say $B$ is $\Theta$-odd if $\Theta B \Theta^{-1} = B^*$.

\end{itemize}

\subsection*{Acknowledgements} We gratefully acknowledge support from the National Science Foundation DMS-2054589 and CAREER DMS-2439949 (AD) and the Pacific Institute for the Mathematical Sciences (XZ). JS was supported in part by NSF grant DMS-2510207. The contents of this work are solely the responsibility of the authors and do not necessarily represent the official views of PIMS.

\newpage

\section[General theory of Z2-index]{\texorpdfstring{General theory of $\Z_2$-index}{General theory of Z2-index}}\label{sec-2}
This section reviews the $\Z_2$ index for an admissible pair consisting of a unitary operator $\UU$ and a projection $\PP$. We introduce two equivalent definitions and establish compactness and norm stability.

Let $\Theta$ be a time-reversal map.
\begin{definition}
We say a pair of a unitary $\UU$ and a projection $\PP$, denoted by $(\UU,\PP)$, is \textit{$\Theta$-admissible} if $[\UU,\PP]$ is trace-class\footnote{Compact is enough for $\Z_2$-index to be well-defined, but trace-class is more convenient for the quantitative estimate later. }, $\UU$ is $\Theta$-odd, $\PP$ is $\Theta$-invariant, i.e.
    \begin{align}\label{eq-2e}
        \Theta \PP \Theta^{-1} = \PP, \qquad \Theta \UU \Theta^{-1} = \UU^*.
    \end{align}
\end{definition}
One can define the $\Z_2$-index for such a $\Theta$-admissible pair in two equivalent ways:
\begin{lemma}\label{lemma-2a}
Assume $(\UU, \PP )$ is $\Theta$-admissible. Then
    \begin{align}
        \II_1(\UU, \PP ):=& \dim\ker(\PP ^\perp + \PP \UU \PP ) \mod 2\label{eqdef-PUP}\\
        =& \dim\ker(\1 - \PP(1 -\UU) \PP)  \mod 2\\
        \II_2(\UU , \PP ):=& \dim\ker(\PP ^\perp + \UU ^*\PP \UU ) \mod 2\\
        =& \dim\ker(\1 - \UU^*[\UU, \PP]) \mod 2 \label{eqdef-UPU}
    \end{align}
are both well-defined and equal $\II_1(\UU , \PP ) = \II_2(\UU , \PP )$. We denote them by $\II(\UU, \PP)$.
\end{lemma}

\begin{proof}
    \textbf{1.} First note that the second equalities in both definitions are immediate. In particular, $\II_2(\UU, \PP)$ is just the parity of the multiplicity of the eigenvalue $1$ of the operator $\UU^*[\UU, \PP]$. Since $\UU^*[\UU , \PP ]$ is compact, the multiplicity of $1$ is indeed finite; hence \eqref{eqdef-UPU} is well-defined.

    \textbf{2.} Then it remains to show $\ker(\PP ^\perp + \PP \UU \PP ) = \ker(\PP ^\perp + \UU ^* \PP \UU )$. Indeed,
    \[
        (\PP ^\perp +  \UU ^* \PP  \UU )\psi = 0  ~ \Leftrightarrow ~ \begin{cases}
        \PP ^\perp \psi = 0\\
        \UU ^* \PP  \UU  \psi = 0
        \end{cases} \Leftrightarrow ~\begin{cases}
            \PP ^\perp \psi = 0 \\
          \PP  \UU  \PP  \psi = 0
        \end{cases}\Leftrightarrow ~(\PP ^\perp + \PP  \UU  \PP ) \psi =0,
\]
where the first implication follows from the non-negativity of $\PP$ and $\UU^*\PP\UU$, while the last reverse implication follows from $\Range(\PP ) \perp \Range (\PP ^\perp)$ because $\PP$ is self-adjoint.

This completes the proof.
\end{proof}
\begin{remark}\label{rmk-2a}
Formula \eqref{eqdef-PUP} is the Fredholm-theoretic form of the index: since $[\UU,\PP]$ is compact, $\PP^\perp+\PP\UU\PP$ is Fredholm, with approximate inverse $\PP^\perp+\PP\UU^*\PP$. Its ordinary Fredholm index vanishes because of the odd time-reversal symmetry, but the parity of its kernel is the non-trivial $\Z_2$ index introduced by Atiyah and Singer \cite{Atiyah1969}; see also \cite{SB_2015,FSSWY20}. Formula \eqref{eqdef-UPU} rewrites the same index as the parity of the multiplicity of the eigenvalue $1$ of the compact operator $\AAA=\UU^*[\UU,\PP]$. This analytic form is useful for the local-determinacy estimate below.
\end{remark}

Since \eqref{eqdef-PUP} contributes more towards the algebraic properties of $\Z_2$-index using algebraic methods ($K$-theory), while \eqref{eqdef-UPU} contributes more towards the analysis properties of the $\Z_2$-index using analytic methods (functional analysis), we call \eqref{eqdef-PUP} and \eqref{eqdef-UPU} \textit{algebraic and analytic definitions of $\Z_2$-index}, respectively.

Finally, we review some basic properties of time-reversal map $\Theta$: Recall that antiunitary means $\lr{\Theta x, \Theta y} = \lr{y, x}$. Note also that, unlike linear operators,  the adjoint of an antilinear operator is defined by $\lr{x, \Theta^* y} = \overline{\lr{\Theta x, y}}$. As a result,
\[
\langle x, \Theta^*y \rangle = \langle y, \Theta x\rangle = \langle \Theta^2x, \Theta y\rangle = \langle x, -\Theta y \rangle
\]
where we used $\Theta^2 = -1$. This means $\Theta^* = -\Theta$. Hence
\begin{align}\label{eq-2f}
    \Theta\Theta^* = \Theta^*\Theta = -\Theta^2 = \1, \text{~~and~~}\Theta^{-1} = \Theta^* = -\Theta.
\end{align}
Moreover, if $S$ is a linear operator, then $S\Theta$ is antilinear, with the adjoint $(S\Theta)^* = \Theta^* S^*$.

\subsection{Compact stability}
\begin{proposition}[Homotopy and compact stability]\label{prop-2a}
Let $T_t$, $0\leq t\leq 1$, be a norm-continuous path of $\Theta$-odd Fredholm operators. Then $\dim\ker T_t\mod 2$ is independent of $t$. In particular, the $\Z_2$ index is invariant under $\Theta$-odd norm-continuous Fredholm homotopies and under compact perturbations through $\Theta$-odd Fredholm operators.
\end{proposition}
\begin{proof}
This is the standard homotopy invariance of the mod-two Fredholm index for odd-symmetric Fredholm operators; see \cite{Atiyah1969,SB_2015,FSSWY20}. 
\end{proof}
\subsection{Norm stability} We first introduce several notations from \cite{ASS94}. Let
\begin{equation}
\QQ := \UU ^* \PP  \UU , \qquad
   \AAA  := \PP  - \QQ = \PP - \UU ^*\PP  \UU = \UU^*[\UU, \PP] , \qquad \BBB := \PP ^\perp - \QQ,
\end{equation}
Using these notations, \eqref{eqdef-UPU} can be written as
\begin{align}
    \II(\UU , \PP ) \equiv \dim \ker (\Id -\UU^*[\UU, \PP]) \equiv \dim \ker (\Id -\AAA).
\end{align}
This is saying that $\II(\UU, \PP)$ is nothing but the parity of the multiplicity of eigenvalue $1$ of $\AAA$.

Without further restrictions, perturbations of an operator in the operator norm might lead to unpredictable variation of the multiplicity of a fixed eigenvalue. For instance, let $B_t = t e_1\otimes e_1$. Then $B_t$ is continuous in $t$ in the operator norm topology. But the multiplicity (and the parity) of the eigenvalue $1$ of $B_t$ is discontinuous at $t = 1$.

Luckily, when $(\UU, \PP)$ is $\Theta$-admissible, any eigenvalue of $\AAA$ that is not $\pm 1$ will have even multiplicities. Hence, perturbations that preserve $\Theta$-admissibility will not change the parity of the multiplicity of eigenvalue $1$, i.e., the $\Z_2$-index is preserved. This is the main idea of the norm stability result, \Cref{prop-2b}, below:
\begin{proposition}[Local determinacy]\label{prop-2b}
    Let $(\UU _1, \PP _1)$, $(\UU _2, \PP _2)$ be two $\Theta$-admissible pairs, such that $\AAA _1 = \UU _1^* [\UU _1,\PP _1]$ and $\AAA _2 = \UU _2^* [\UU _2,\PP _2]$ satisfy:
    \begin{equation}
        \label{eq-3a}
    \|  \AAA _1 - \AAA _2\| _{\Tr}<\frac{2^{-8}}{ (\|  \AAA _1 \| _{\HS}^2 + \|  \AAA _2 \| _{\HS}^2 + 1)^2}.
    \end{equation}
    Then $\II(\UU _1, \PP _1) = \II(\UU _2, \PP _2)$.
\end{proposition}

 Before we prove the proposition, we first recall the even-multiplicity property away from $\pm1$.

\begin{lemma}\label{lemma-even-multiplicity} Assume $(\UU ,\PP )$ is $\Theta$-admissible. Let $m_\lambda(\AAA)$ denote the multiplicity of eigenvalue $\lambda$ of $\AAA = \UU^*[\UU, \PP]$. Then $m_\lambda(\AAA)$ is even for all $\lambda\neq \pm 1$.
\end{lemma}

\begin{proof} \textbf{0.} By \eqref{eq-2e}, one can easily check
\begin{equation}
(\Theta \UU ) \AAA  = -\AAA  (\Theta \UU ), \qquad (\Theta \UU ) \BBB = \BBB (\Theta \UU ).
\end{equation}
Moreover, direct computations give $\AAA ^2+\BBB^2 = \Id$ and $\AAA \BBB +\BBB\AAA  = 0$. As a consequence, $\|  \AAA \| ^2 = \|  \AAA ^2 \|  \leq 1$ and $\Spec(\AAA ) \subset [-1,1]$. In particular, any eigenvalue of $\AAA $, different from $\pm 1$, is in $(-1,1)$.

\textbf{1.} Assume now that $\lambda \neq \pm 1$ is an eigenvalue of $\AAA $. Let $\EE$ denote the corresponding eigenspace of $\AAA $ at energy $\lambda$. Since $\AAA $ is compact, $\EE$ is finite dimensional as long as $\lambda \neq 0$. Let $J = \Theta \UU  \BBB$.  Given any $v\neq 0 \in \EE$, we claim that $\Span\{v, Jv\}$ is a closed subspace of $\EE$ with dimension $2$. Indeed, we first note that
\[
\AAA Jv = \AAA \Theta \UU  \BBB v = -\Theta \UU  \AAA \BBB v = \Theta \UU  \BBB  \AAA v = J\AAA v = \lambda Jv.
\]
Hence $Jv\in \EE$. Moreover, we note
\begin{align*}
    \langle v, Jv\rangle = \langle  v, \Theta \UU  \BBB  v\rangle = \langle v, \BBB \UU ^*\Theta^* v \rangle = \langle v, -\BBB \Theta \UU  v \rangle = -\langle v, \Theta \UU  \BBB  v\rangle = -\langle v, Jv\rangle = 0,
\\
\|  Jv\| ^2 = \langle Jv, Jv\rangle = \langle \Theta \UU  \BBB v, \Theta \UU  \BBB v\rangle = \langle \UU \BBB v, \UU \BBB v \rangle = \langle v, \BBB ^2 v \rangle = (1 - \lambda^2)\|  v\| ^2 \neq 0.
\end{align*}
As a result, $v$ and $Jv$ are linearly independent. Hence $\Span\{v, Jv\}$ is a closed two dimensional subspace of $\EE$.

\textbf{2.} For any $w \perp \Span\{v, Jv\}$, we have $\Span \{w, Jw\} \perp \Span\{v, Jv\}$. In fact, we note
\[
J^* = \BBB \UU ^*\Theta^* = -\BBB \UU ^* \Theta = -\BBB  \Theta \UU  = -\Theta \UU  \BBB  = -J, \text{~~and~~} J^* J = \BBB \UU ^* \Theta^* \Theta \UU  \BBB  = \BBB ^2.
\]
As a result,
\[
\langle Jw, v\rangle = \langle J^*v, w\rangle = \langle -Jv, w\rangle = 0,
\]
\[
\langle Jw, Jv\rangle = \langle J^*Jv, w\rangle = \langle \BBB ^2v, w\rangle = (1 - \lambda^2)\langle v, w\rangle = 0.
\]

\textbf{3.} We now pick inductively $v_1, \dots, v_k \in \EE$ such that
\begin{equation}
    0 \neq v_k  \perp   \bigoplus_{i = 1}^{k -1}\Span\{v_i, Jv_i\},
\end{equation}
as long as
\begin{equation}
    \EE  \neq \bigoplus_{i = 1}^{k -1}\Span\{v_i, Jv_i\}.
\end{equation}
Because $\AAA $ is compact, $\dim(\EE)$ is finite and hence this procedure ends after, say, $n$ steps. This implies that $\dim \EE = 2n$, and this completes the proof.
\end{proof}

\begin{corollary}\label{cor-spectral-count}
    If $(\UU , \PP )$ is admissible and $0 < a < 1 < b$, then modulo 2:
    \[
    \II(\UU , \PP ) \equiv \Tr\big( \1_{(a, b)}(\AAA ) \big).
    \]
\end{corollary}

\begin{proof}
    Since $\AAA $ is compact, it has finitely many eigenvalues in $(a, b)$ for $0<a<1<b$, hence $\1_{(a, b)}(\AAA )$ has finite rank.
    By \Cref{lemma-even-multiplicity}, $m_\lambda(\AAA ) \equiv 0 \mod 2$ for $a<\lambda\neq 1<b$. Hence,
    \[
    \Tr(\1_{(a, b)}(\AAA )) = \sum_{\lambda \in (a, b)} m_\lambda(\AAA ) \equiv m_1(\AAA ) \equiv \dim \ker(\1 - \AAA ).
    \]
    This completes the proof.
\end{proof}

Now we are ready to prove \Cref{prop-2b}:
\begin{proof}[Proof of \Cref{prop-2b}] 1. Define
\begin{align}
    r := \frac{2^{-4}}{\|  \AAA _1\| _{\HS}^2 +\|  \AAA _2\| _{\HS}^2 + 1},
\end{align}
so that the assumption \eqref{eq-3a} can be rewritten as $\|  \AAA _1 - \AAA _2\| _{\Tr} < r^2$. We first claim that there exists $a\in (1/2,1)$  such that
\begin{equation}\label{eq-1m}
    d(a, \Sigma(\AAA _1) \cup \Sigma(\AAA _2)) \geq r.
\end{equation}
Indeed, because $\AAA _1, \AAA _2$ are compact, their spectrum away from $0$ consist of eigenvalues with finite multiplicity, and
   \[
   \|  \AAA _i\| _{\HS}^2  \geq \sum_{\lambda_k\in[1/2,1]} |\lambda_k|^2\geq \frac{1}{4}\cdot \#\big\{\Sigma(\AAA _i)\cap [1/2,1]\big\}, \qquad i = 1,2.
   \]
   As a result,
   \begin{equation}
       \label{eq-3b}
   \#\big\{\big(\Sigma(\AAA _1) \cup \Sigma(\AAA _2)\big) \cap [1/2, 1]\big\} \leq 4 (\|  \AAA _1\| _{\HS}^2 + \|  \AAA _2 \| _{\HS}^2).
   \end{equation}
   Now evenly cut $[1/2,1]$ in $N + 1$ consecutive intervals, where $N$ is the integer closest to $4 (\|  \AAA _1\| _{\HS}^2 + \|  \AAA _2 \| _{\HS}^2)+1$. By the pigeonhole principle, one of them must not intersect $\Sigma(\AAA _1) \cup \Sigma(\AAA _2)$. If $a$ is the center of this interval, then
   \begin{align}
       d\big( a, \Sigma(\AAA _1) \cup \Sigma(\AAA _2) \big) \geq \dfrac{1}{2(N+1)} \geq r,
   \end{align}
   as claimed. This proves \eqref{eq-1m}.

2.  Let $a$ satisfying \eqref{eq-1m} and $b=a+1$. By \Cref{cor-spectral-count}, $\II(\UU _i, \PP _i) \equiv \Tr(\1_{(a, b)}(\AAA _i))$. Let $\Gamma$ be the rectangle with vertices $(a, \pm \tfrac{1}{2})$, $(b, \pm\tfrac{1}{2})$. We see that
    \begin{align}
     \1_{(a, b)}(\AAA _1) -  \1_{(a,b)}(\AAA _2) &= \frac{1}{2\pi i } \oint_{\Gamma} (z - \AAA _1)^{-1} - (z - \AAA _2)^{-1} dz\\
     & = \frac{1}{2 \pi i}\oint_\Gamma (z - \AAA _1)^{-1}(\AAA _1 - \AAA _2)(z - \AAA _2)^{-1} dz
    \end{align}
Taking trace norm on both sides, and using our assumption $\|  \AAA _1 - \AAA _2\| _{\Tr} < r^2$, we obtain
   \begin{align}
      \|  \1_{(a, b)}(\AAA _1) - \1_{(a, b)}(\AAA _2) \| _{\Tr}&\leq \frac{1}{2\pi}\oint_\Gamma \|  (z - \AAA _1)^{-1} \|  \|  \AAA _1 - \AAA _2 \| _{\Tr}  \|  (z - \AAA _2)^{-1} \|  |dz|\\
       &\leq \frac{|\Gamma|}{2\pi}\frac{1}{d(a, \Sigma(\AAA _1))}\cdot\|  \AAA _1 - \AAA _2\| _{\Tr}\cdot \frac{1}{d(a, \Sigma(\AAA _2))}\\
       &\leq \frac{4}{2\pi}\frac{r^2}{r^2} <1.
   \end{align}
Because $\Tr(\1_{(a, b)}(\AAA _i))$ are integers, we conclude that they are equal, and hence $\II(\UU _1, \PP _1) \equiv \II(\UU _2, \PP _2)$. \end{proof}

\newpage
\section[Z2-index for Hamiltonians]{\texorpdfstring{$\Z_2$-index for Hamiltonians}{Z2-index for Hamiltonians}}\label{sec-3}
In the previous section, we discussed the general theory of $\Z_2$-index $\II(\UU, \PP)$ for a pair of $\Theta$-admissible unitary and projection $(\UU, \PP)$. In this section, we apply the general theory to more concrete objects -- the geometric bulk index $\II_{\Omega, W}(H)$, the bulk index $\II(H)$, and the edge index $\NN(H_e)$ -- see \eqref{eqdef-geobulk}, \eqref{eqdef-bulk}, and \eqref{eqdef-edge}. Indeed,
\begin{align}
  &\II_{\Omega, W}(H) = \II(U_\Omega, \1_W), \qquad U_\Omega = e^{-2\pi i \1_{\Omega} \Pi \1_{\Omega} }, \quad \Pi = \1_{(-\infty, \inf \GG)}(H),\\
  &\II(H) =  \II(U_{\Hh_2}, \1_{\Hh_1}), \\
 &\NN(H_e) = \II(e^{-2\pi i g(H_e)}, \1_{W}).
\end{align}
\subsection{Kernel estimates}
We review some definitions and results from \cite{DZ24}:

\begin{definition}\label{def-1} We say that a self-adjoint operator $A$ on $\ell^2(\Z^2)$ is \textit{exponentially-short-range (ESR)} if its kernel satisfies, for some $\nu >0$,
\begin{equation}\label{eq-9z}
    |A(x, y)|\leq \nu^{-1}e^{-2\nu |x-y|}, \qquad x, y\in \Z^2.
\end{equation}
Moreover, we say that $A$ decays exponentially away from $S \subset \R^2$ if, for some $\nu >0$,
\begin{equation}
    |A(x, y)|\leq \nu^{-1}e^{-2\nu |x-y|-\nu d(x,S) - \nu d(y,S)}, \qquad x, y\in \Z^2.
\end{equation}
\end{definition}

\begin{definition}
    An operator $A$ on $\ell^2(\Z^2,\C^m)$ is \textit{polynomially-short-range (PSR)} if for any $N>0$, there is $C_N>0$ such that
\[
    |A(x, y)|\leq C_N (|x-y|+1)^{-N}, \qquad x, y\in \Z^2.
\]
Moreover, we say that $A$ decays polynomially away from $S \subset \R^2$ if for any $N>0$, there is $C_N>0$ such that
\begin{equation}
    |A(x, y)|\leq C_N (|x-y| + d(x,S) + d(y,S) +1)^{-N}, \qquad x, y\in \Z^2.
\end{equation}
\end{definition}

Note that ESR (respectively exponential decay) is a stronger condition than PSR (respectively polynomial decay). Furthermore, smooth functional calculus of ESR operators is at least PSR\footnote{one could expect ESR for analytic functional calculus on the spectrum of the ESR operator, e.g. $\Pi = \1_{(-\infty, \MG)}(H)$, as in \cite[Lemma 3.1]{DZ23}}:
\begin{lemma}\label{lemma-functional-calculus}
    Let $f \in C^\infty(\R)$. Assume that $A,B$ are two ESR self-adjoint operators such that $A-B$ decays exponentially away from $S\subset\Z^2$. Then $f(A)$ and $f(B)$ are PSR, and $f(A)-f(B)$ decays polynomially away from $S$.
\end{lemma}

The proof decomposes $f(A)$ as an integral of the resolvents $(z-A)^{-1}$ using a Helffer-Sj\"ostrand formula. It then estimates the resolvent $(z-A)^{-1}$ via a Combes--Thomas estimate. We refer to the proof of \cite[Equation (3.9)]{DZ24} for details. Similarly, it was proved in \cite[Equation (3.10)]{DZ24} that

\begin{lemma}\label{lem-2b}
    $g(H_e) - \1_{\Omega}g(H_+)\1_\Omega - \1_{\Omega^c} g(H_-) \1_{\Omega^c}$ is PSR and decays polynomially away from $\p \Omega$ for $g$ in \Cref{def-edge_index}.
\end{lemma}

\subsection{Well-definedness}
\begin{proposition}\label{prop:geometric index well defined}
    If $\Omega,W$ are transverse and $X$ is a $\Theta$-invariant PSR self-adjoint operator such that $X^2-X$ decays polynomially away from $\p \Omega$, then $[e^{-2\pi i X}, \1_W]$ is compact and $\II(e^{-2\pi i X}, \1_W)$ is well-defined.
\end{proposition}

\begin{proof}
It is easy to check that $\1_W$ is $\Theta$-invariant and $e^{-2\pi i X}$ is $\Theta$-odd because $X$ is $\Theta$-invariant. By \Cref{lemma-2a}, it remains to prove that $\left[e^{-2i\pi X}, \1_W\right]$ is trace-class.

Because $e^{-2i\pi x} - 1$ vanishes at both $0$ and $1$, there exists an entire function $h$ such that
\begin{equation}
    e^{-2i\pi x}-1 = h(x)(x^2-x).
\end{equation}
Therefore, we have:
\begin{equation}
    \left[\1_W,e^{-2 i\pi X}\right] = \left[\1_W,h(X)\right] (X^2-X) +  h(X)\left[\1_W,X^2-X\right].
\end{equation}

By \Cref{lem-2b}, $h(X)$ is PSR. By \cite[Lemma 2.4]{DZ24}, $\left[\1_W,h(X)\right]$ decays exponentially away from $\p W$. By assumption, $X^2-X$ decays polynomially away from $\p \Omega$. By \cite[Corollary 2.4]{DZ24}, $\left[\1_W,h(X)\right] (X^2-X)$ is trace-class with its kernel decaying polynomially away from both $\p \Omega$ and $\p W$. Likewise, $h(X)\left[\1_W,(X^2-X)\right]$, and hence $[\1_W, e^{-2\pi i X}]$, are trace-class with a kernel that decays polynomially away from $\p \Omega$ and $\p W$. Hence $\II(e^{-2\pi i X}, \1_W)$ is well-defined.

This completes the proof.

\end{proof}
\begin{corollary}\label{cor-indices-well-defined}
    If $\Omega, W$ are transverse, $[U_\Omega, \1_W]$, $[e^{-2\pi i g(H_e)}, \1_W]$ are compact and $\II_{\Omega, W}(H)$, $\II(H)$, $\NN(H_e)$ in \eqref{eqdef-geobulk}, \eqref{eqdef-bulk}, \eqref{eqdef-edge} are well-defined.
\end{corollary}
\begin{proof}
    By \Cref{prop:geometric index well defined}, it remains to prove for $X_1 = \1_{\Omega} \Pi \1_{\Omega}$ and $X_2 = g(H_e)$, we have $X_j^2 - X_j$ decays polynomially away from $\p \Omega$, $j = 1,2$. Indeed,
\[
X_1^2 - X_1 = (\1_\Omega \Pi \1_\Omega)^2 - \1_\Omega \Pi \1_\Omega = \1_\Omega [\Pi,  \1_\Omega] \1_\Omega.
\]
By \cite[Lemma 2.1]{DZ24}, $\Pi$ is PSR, by \cite[Lemma 2.1]{DZ24}, $[\Pi, \1_\Omega]$ decays polynomially away from $\p \Omega$; so does $X_1^2 - X_1$.

Meanwhile, $g\in C^\infty$ is $1$ on $(-\infty, \inf \MG)$ and $0$ on $[\sup \MG, +\infty)$, hence $g^2(x) - g(x) = \phi(x)\in C^\infty$ where $\supp (\phi) \subset \MG$. There is $\rho(x) = \int_0^x \phi(s) ds$ such that $\rho\in C^\infty$, $\rho(x) = 1$ on $[\sup \MG, +\infty)$, $\rho(x) = 0$ on $(-\infty, \inf \MG]$. Hence by \cite[Prop. 3, (3.10)]{DZ24}, $X_2^2 - X_2 = \phi(H_e) = \rho'(H_e)$ decays polynomially away from $\p \Omega$.

\end{proof}
\begin{remark}\label{rmk-3a}
    Note that a byproduct of this proof is when $\Omega$, $W$ are transverse, $[U_\Omega, \1_W]$ and $\AAA_{\Omega, W}:= U_\Omega^*[U_\Omega, \1_W]$ decays polynomially away from both $\p \Omega$ and $\p W$.
\end{remark}
\begin{lemma}\label{lemma-3c}
    Let $\Omega$, $W$ be transverse. If $X, Y$ are PSR, $\Theta$-invariant self-adjoint operators such that both $X^2 - X$ and $X - Y$ decay polynomially away from $\p \Omega$, then $Y^2 - Y$ decays polynomially away from $\p \Omega$ and
    \begin{align}
        \label{eq-1c}
    \II(e^{-2\pi i Y}, \1_W) = \II(e^{-2\pi i X}, \1_W).
    \end{align}
\end{lemma}
\begin{proof}
    Define $X(t) := X + t(Y-X)$. We have
\begin{align}
        X(t)^2-X(t)
        = t^2\br{Y-X}^2+t\br{(Y-X)X+X(Y-X)}+X^2-X-t(Y-X)\,.
    \end{align}
    As a result, for every $t$, $X(t)^2-X(t)$ decays away from $\partial \Omega$; hence $\II(e^{-2\pi i X(t)}, \1_W)$ is well-defined for $t\in [0,1]$ by \Cref{prop:geometric index well defined}; in particular, it is well-defined at $X(1)=Y$. By \Cref{prop-2a} , $\1_{W^c} + \1_W e^{-2\pi i X(t)} \1_W$ defines a homotopy of $\Theta$-odd Fredholm operators. The $\Z_2$ index is constant along such paths and hence indices at $t=0$ and $t=1$ agree. This precisely proves \eqref{eq-1c}.
\end{proof}

This leads to another equivalent definition of $\II_{\Omega, W}(H)$:

\begin{corollary}\label{cor-3d}
    Let $\Omega$, $W$ be transverse. Recall $U_\Omega = e^{-2\pi i \1_{\Omega} \Pi \1_\Omega}$. Let $V_{\Omega}= e^{-2\pi i \Pi \1_\Omega \Pi}$. Then
    \begin{align}
        \label{eq-3c}
    \II_{\Omega, W}(H):= \II(U_\Omega, \1_W) = \II(V_\Omega, \1_W).
    \end{align}
\end{corollary}

\begin{proof}
    Note that
    \[
     \1_\Omega \Pi \1_\Omega  - \Pi \1_\Omega \Pi = \1_\Omega[\Pi, \1_\Omega] - [\Pi, \1_\Omega] \Pi.
    \]
    By \cite[Lemma 2.1]{DZ24}, $[\Pi, \1_\Omega]$ decays polynomially away from $\p \Omega$. By \Cref{lemma-3c} applying to $X = \1_\Omega \Pi \1_\Omega$ and $Y = \Pi \1_\Omega \Pi$, we obtain \eqref{eq-3c}.
\end{proof}

By \Cref{cor-3d} and \Cref{lemma-2a}, we have four equivalent definitions and their variations for the bulk index:
\begin{align}
    \II_{\Omega, W}(H) = \II(U_\Omega, \1_W) &\equiv \dim\ker(\1_{W^c} + \1_WU_\Omega \1_W) = \dim\ker(\1 - \1_W(\1 - U_\Omega)\1_W) \label{eqdef-3a}\\
    &\equiv \dim\ker(\1_{W^c} + U_\Omega^*\1_W U_\Omega) = \dim\ker(\1 - U_\Omega^*[U_\Omega, \1_W]) \label{eqdef-3b}\\
    =  \II(V_\Omega, \1_W) &\equiv \dim\ker(\1_{W^c} + \1_W V_\Omega \1_W) \label{eqdef-3c}\\
    &\equiv \dim\ker(\1_{W^c} + V_\Omega^*\1_W V_\Omega),\label{eqdef-3d}
\end{align}
where $U_\Omega = e^{-2\pi i \1_{\Omega} \Pi \1_\Omega}$, $V_{\Omega}= e^{-2\pi i \Pi \1_\Omega \Pi}$. As we discussed in \Cref{rmk-2a},  \Cref{eqdef-3a}, \Cref{eqdef-3c} are Fredholm, hence invariant under $\Theta$-odd compact perturbations; \Cref{eqdef-3b}, \Cref{eqdef-3d} are involved with local operators.
\subsection{Complement}\label{subsec-3.3}

\begin{lemma}\label{lemma-3e} Let $\Omega$, $W$ be transverse. We have
    \begin{equation}
        \label{eq-4a}
     \II_{\Omega, W}(H) \equiv \II_{\Omega^c,W}(H) \equiv \II_{\Omega,W^c}(H).
    \end{equation}
\end{lemma}

\begin{proof} \textbf{1. } We first show $ \II_{\Omega,W}(H) \equiv \II_{\Omega, W^c}(H)$. Recall $\1_W \Theta = \Theta \1_W$, $\Theta U_\Omega = U_\Omega^* \Theta$, $(\Theta U_\Omega)^* = U_\Omega^* \Theta^*$. Therefore,
\[
\Theta U_\Omega(\1_{W^c} + U_\Omega^* \1_W U_\Omega) (\Theta U_\Omega)^* = \Theta (U_\Omega \1_{W^c} U_\Omega^* + \1_W) \Theta^* = \1_W + U_\Omega^* \1_{W^c} U_\Omega.
\]
As a result, by definition \eqref{eqdef-3b},
\[
\II_{\Omega, W}(H) \equiv \dim \ker(\1_{W^c} + U_\Omega^* \1_W U_\Omega) = \dim \ker(\1_W + U_\Omega^* \1_{W^c} U_\Omega) \equiv \II_{\Omega, W^c}(H).
\]

\textbf{2.} We now show $\II_{\Omega,W}(H) = \II_{\Omega^c,W}(H)$ using definition \eqref{eqdef-3c}. Note
\begin{align}
    V_{\Omega^c} = e^{2\pi i\Pi \1_{\Omega^c} \Pi} = e^{2\pi i \Pi - 2\pi i \Pi \1_\Omega \Pi} = e^{2\pi i \Pi}e^{-2\pi i \Pi \1_\Omega\Pi} = e^{-2\pi i \Pi \1_{\Omega} \Pi} = \Theta e^{2\pi i \Pi \1_\Omega \Pi} \Theta^* = \Theta V_\Omega \Theta^*.
\end{align}
It follows that $\Theta (\1_{W^c} + \1_W V_{\Omega} \1_W) \Theta^* = \1_{W^c} + \1_W V_{\Omega^c} \1_W$. Therefore, by definition \eqref{eqdef-3c}:
\[
\II_{\Omega, W}(H) \equiv \dim\ker( \1_{W^c} + \1_W V_{\Omega} \1_W) = \dim \ker(\1_{W^c} + \1_W V_{\Omega^c} \1_W) \equiv \II_{\Omega^c, W}(H).
\]
This completes the proof. \end{proof}

\subsection{Additivity}
\begin{lemma}[Additivity in $\Omega$]\label{lemma-3f} Let $\Omega_1, \Omega_2, W \subset \R^2$ be such that $(\Omega_1,W)$, $(\Omega_2,W)$ are transverse, and $\Omega_1\cap  \Omega_2 = \emptyset$. Assume that $\Pi_1, \Pi_2$ are two PSR, $\Theta$-invariant projectors. Then
\[
\II(U_1U_2, \1_W) = \II(U_1, \1_W) + \II(U_2, \1_W)
\]
where $U_j = e^{-2\pi i \1_{\Omega_j}\Pi_j \1_{\Omega_j}}$.

In particular, when $\Pi_1 = \Pi_2 = \Pi$, we have
\[
\II(U_{\Omega_1 \sqcup  \Omega_2}, \1_W) = \II(U_{\Omega_1}, \1_W) + \II(U_{\Omega_2}, \1_W).
\]
\end{lemma}
\begin{proof}
To prove the lemma, it suffices to justify that $\EE = \EE_1\oplus \EE_2$, where
\begin{equation}
    \EE_j \de \ker \big( \1_{W^c} + \1_W U_j \1_W) \qquad \EE \de \ker \big( \1_{W^c} + \1_W U_1 U_2 \1_W \big).
\end{equation}
By the definition of a direct sum, it remains to show: 1. $\EE_1\cap \EE_2 = \{0\}$; 2. every $v\in \EE$ can be written as $v_1 + v_2$, where $v_j\in \EE_j$.

\textbf{0.} First note
\begin{align}\label{eq-2a}
    \1 - U_j = \1 - e^{-2\pi i \1_{\Omega_j} \Pi_j \1_{\Omega_j}} = \sum_{k = 1}^\infty (-2\pi i \1_{\Omega_j}\Pi_j \1_{\Omega_j})^k =: \1_{\Omega_j} \QQ_j\1_{\Omega_j}
\end{align}
for some operator $\QQ_j$. In particular, since $\Omega_1 \cap \Omega_2 = \emptyset$, we have
\begin{align}
    \label{eq-2d}
\1 - U_1 - U_2 + U_1U_2 = (1 - U_1)(1 - U_2) = \1_{\Omega_1} \QQ_1 \1_{\Omega_1}\1_{\Omega_2} \QQ_2 \1_{\Omega_2}  = 0
\end{align}

\textbf{1.} We show $\EE_1 \cap \EE_2 = \{0\}$ here. Note that
\begin{equation}\label{eq-1e}
    v \in \EE_j \ \Leftrightarrow \ \systeme{ \1_{W^c} v = 0 \\ \1_W U_j \1_W v = 0} \ \Leftrightarrow \ \systeme{ \1_{W^c} v = 0 \\ \1_W U_j v = 0} \ \Leftrightarrow \ \systeme{ \1_{W^c} v = 0 \\ \1_W (\1-U_j) v =\1_W\1_{\Omega_j} \QQ_j\1_{\Omega_j}v = v}.
\end{equation}
The last two equations imply that $v$ is supported in $W$ and $\Omega_j$. In particular, if $v\in \EE_1\cap \EE_2$, then $v$ is supported in $\Omega_1\cap \Omega_2 = \emptyset$; hence $v = 0$. Thus $\EE_1 \cap \EE_2 = \{0\}$.

\textbf{2.} Meanwhile if $v\in \EE$, we claim $v$ is supported on $\Omega_1\sqcup \Omega_2$. Indeed, by \eqref{eq-2d},
\[
\begin{aligned}
v\in \EE
&\Leftrightarrow \systeme{\1_{W^c}v = 0\\ \1_W U_1U_2 v = 0} \\
&\Leftrightarrow \systeme{\1_{W^c}v = 0 \\ \1_W(U_1 + U_2 - \1)v = 0} \\
&\Leftrightarrow \systeme{\1_{W^c} v = 0\\ \1_W(\1 - U_1)v + \1_W(\1 - U_2)v = v}.
\end{aligned}
\]
Since $\1 - U_j = \1_{\Omega_j}\QQ_j\1_{\Omega_j}$, the last two equations imply that $v$ is supported on $W$ and $\Omega_1 \sqcup \Omega_2$.

As a result, we can write $v = \1_{\Omega_1}v + \1_{\Omega_2}v$. Now we claim $v_j := \1_{\Omega_j}v\in \EE_j$. Indeed, apply $\1_{\Omega_1}$ to the last two equations above, since $\1_{\Omega_1} (\1 - U_2) = 0$ by \eqref{eq-2a}, we obtain
\[
    \1_{\Omega_1}\1_{W^c} v = \1_{W^c}\1_{\Omega_1}v = \1_{W^c} v_1 = 0
\]
    and
\[
  v_1 = \1_{\Omega_1}v = \1_{\Omega_1} \1_W(\1 - U_1) v =  \1_W\1_{\Omega_1} (1 - U_1)v = \1_W (\1 - U_1) \1_{\Omega_1} v = \1_W (\1 - U_1) v_1
\]
where we used $\1_{\Omega_j}(\1 - U_j) = (\1 - U_j) \1_{\Omega_j}$ by \eqref{eq-2a}. This implies $\1_W U_1 v_1 = 0$. Together with $\1_{W^c}v_1 = 0$, we obtain $v_1 \in \EE_1$. A similar argument proves $v_2\in \EE_2$.

This completes the proof.
\end{proof}

\begin{lemma}[Additivity in $W$]\label{lemma-3g}  Let $W_1, W_2, \Omega$ be open subsets of $\R^2$ such that $(\Omega, W_1)$ and $(\Omega,W_2)$ are transverse, and $W_1\cap W_2 = \emptyset$. Then $(\Omega, W_1\sqcup W_2)$ is transverse and
\begin{equation}\label{eq-3d}
\II(U_\Omega, \1_{W_1 \sqcup W_2}) \equiv \II(U_\Omega, \1_{W_1}) + \II(U_\Omega, \1_{W_2}).
\end{equation}
\end{lemma}

\begin{proof}
\textbf{0.} Define $W = W_1 \sqcup W_2$. We first show $(\Omega, W)$ is transverse. Note that $\p W = \p W_1 \cup \p W_2$ if $W_1$, $W_2$ are disjoint open sets; therefore
\begin{align}
 \Psi_{\Omega, W}(x) & = 1 + d(x, \p(W_1 \sqcup W_2)) + d(x, \p \Omega)\\
 & =  1 + \min\big\{d(x, \p W_1), d(x, \p W_2)\big\} + d(x, \p \Omega) = \min \big\{\Psi_{\Omega, W_1}(x), \Psi_{\Omega, W_2}(x)\big\}.
\end{align}
Hence $(\Omega, W)$ is transverse.

\textbf{1.} Recall by \Cref{eqdef-3a} and \Cref{eq-2a},
\begin{align}
        \II(U_\Omega, \1_W) &\equiv \dim\ker(\1 - \1_W(1 - U_\Omega)\1_W)\\
        &\equiv \dim\ker(\1 - \1_{W_1}(\1 - U_\Omega) \1_{W_1} - \1_{W_2} (\1 - U_\Omega)\1_{W_2} - \1_{W_1}U_\Omega \1_{W_2} - \1_{W_2} U_\Omega\1_{W_1}).
  \end{align}
Note that $\1_{W_1} U_\Omega \1_{W_2} = \1_{W_1}[U_{\Omega}, \1_{W_2}]$ where $[U_\Omega, \1_{W_2}]$ is $\Theta$-odd and compact. By \Cref{prop-2a},  the index does not change by these perturbations. Hence
\[
\II(U_\Omega, \1_W) \equiv \dim\ker(\1 - \1_{W_1}(1 - U_\Omega)\1_{W_1} - \1_{W_2}(1 - U_\Omega)\1_{W_2} ).
\]

\textbf{2.} Now to prove \Cref{eq-3d}, it remains to show  $\mathcal W = \mathcal W_1 \oplus \mathcal W_2$ where
\[
\mathcal W_j \de \ker(\1 - \1_{W_j}(1 - U_\Omega)\1_{W_j}),  \qquad \mathcal W \de \ker(\1 - \1_{W_1}(1 - U_\Omega)\1_{W_1} - \1_{W_2}(1 - U_\Omega)\1_{W_2} ).\]
That is to prove: 1. $\mathcal W_1\cap \mathcal W_2 = \{0\}$; 2. every $w\in \mathcal W$ can be written as $w_1 + w_2$ for $w_j\in \mathcal W_j$.

\textbf{3.} Now we show $\mathcal W_1 \cap \mathcal W_2 = \{0\}$. Recall by \Cref{eqdef-3a}, $w_j\in \mathcal W_j = \ker(\1 - \1_{W_j} (1 - U_\Omega) \1_{W_j})$.  Hence $w_j = \1_{W_j}(1 - U_\Omega) \1_{W_j}w_j$. This implies $w_j$ is supported on $W_j $. Assume $w\in \mathcal W_1 \cap \mathcal W_2$, then $w$ is supported on $W_1 \cap W_2 = \emptyset$; hence $w = 0$.

\textbf{4.} Now given $w\in \mathcal W$, by definition of $\mathcal W$,
\[
w = \1_{W_1}(1 - U_\Omega)\1_{W_1}w + \1_{W_2}(1 - U_\Omega)\1_{W_2} w =: w_1 + w_2.
\]
By the definition of $w_1$, $w_2$ -- $w_j = \1_{W_j}(1 - U_\Omega) \1_{W_j}$ -- they belong to $\mathcal W_1$, $\mathcal W_2$ respectively.
    This completes the proof. \end{proof}
\subsection{Compact stability} \label{subsec-3.5}

\begin{lemma}\label{lemma-3h}
Let $\Omega$, $W$ be transverse, $C\subset \R^2$ be a compact set. We have
\[
\II_{\Omega, W}(H) = \II_{\Omega, W\cup C}(H) = \II_{\Omega\cup C, W}(H).
\]
\end{lemma}

\begin{proof}
By additivity in $\Omega$ and $W$, it remains to show  $\II_{\Omega, C}(H) = \II_{C, W}(H) = 0$. Indeed, by \Cref{eqdef-3a} and \Cref{eq-2a},
\begin{align}
&\II_{\Omega, C}(H) \equiv \dim\ker( \1 - \1_{C}(1 - U_\Omega) \1_C), \\
&\II_{C, W}(H) \equiv \dim\ker(\1 - \1_W\1_C\mathcal Q_C \1_C \1_W).
\end{align}
Now since $\1_C$ is compact, both $ \1_{C}(1 - U_\Omega) \1_C$ and $\1_W\1_C\mathcal Q_C \1_C \1_W$ are $\Theta$-odd compact perturbations, of $\1$. By \Cref{prop-2a}, $\II_{\Omega, C}(H) \equiv \II_{C, W}(H) \equiv 0$.

This completes the proof.
\end{proof}

\newpage
\section{Proof of Theorem \ref{thm-2}}\label{sec-4}
\begin{proof}
    Recall that $\NN(H_e) = \II(e^{-2\pi i g(H_e)}, \1_W)$. Set
    \[
    X = g(H_e), \quad Y = \1_\Omega g(H_+) \1_\Omega + \1_{\Omega^c} g(H_-) \1_{\Omega^c}.
    \]
    By \Cref{lem-2b}, $X$, $Y$ differs polynomially away from $\p \Omega$. By \Cref{lemma-3c},
    \[
    \NN(H_e) = \II(e^{-2\pi i X}, \1_W) = \II(e^{-2\pi i Y}, \1_W).
    \]

    Furthermore, by definition of $g$ in \eqref{eqdef-edge} and $\MG$ being the bulk spectral gap, we see that
    \[
    g(H_\pm) = \1_{(-\infty, \inf\MG)}(H_\pm)=: \Pi_\pm.
    \]
    Hence we have $Y = \1_\Omega \Pi_+ \1_\Omega + \1_{\Omega^c} \Pi_- \1_{\Omega^c}$. As a result,
    \[
     e^{-2\pi i Y} = U_+U_-, \qquad U_+ = e^{-2\pi i \1_{\Omega}\Pi_+\1_{\Omega}}, \qquad U_-  = e^{-2\pi i \1_{\Omega^c}\Pi_- \1_{\Omega^c}}.
    \]
    By \Cref{lemma-3f},
        \[
        \NN(H_e) = \II(e^{-2\pi i Y}, \1_W) = \II(U_+U_-, \1_W) = \II(U_+, \1_W) + \II(U_-, \1_W).
        \]

    Note that by definition, $ \II(U_+, \1_W) = \II_{\Omega, W}(H_+)$, while
    \[
    \II(U_-, \1_W) = \II_{\Omega^c, W}(H_-) =  \II_{\Omega, W}(H_-) \equiv - \II_{\Omega, W}(H_-),
    \]
    which follows from \Cref{lemma-3e} and $x \equiv -x \mod 2$ for integer $x$. This completes the proof of the theorem:
    \[
    \NN(H_e) \equiv \II_{\Omega, W}(H_+) - \II_{\Omega,  W}(H_-).
    \]
\end{proof}

\newpage
\section{Proof of Theorem \ref{thm-3}}\label{sec-5}
The proof of \Cref{thm-3} follows closely that of \cite[\S 6, \S7]{DZ24} with several $\Z_2$-specific modifications highlighted below. More explicitly, we first prove Theorem \ref{thm-3} for simple $\Omega$ following \cite[\S 6]{DZ24}, where the intersection number $\XX_{\Omega, W}$ takes either $0$, $1$, or $-1$. Then we follow \cite[\S 7]{DZ24} to extend the results to more general transversal sets.

We start with reviewing some basic definitions from \cite[\S 6]{DZ24}:
\begin{definition}[Simple path, loop]\label{def-5a}
    A continuous map $\gamma: \R \to \R^2$ is called a simple path if
    \begin{enumerate}
        \item $\gamma$ is injective and proper (the preimage of a compact set is compact);
        \item There exists a discrete closed set $S\subset \R$ such that $\gamma$ is smooth on $\R\setminus S$.
        \item For all $t\in \R$, the left and right derivative of $\gamma$ at $t$ exist and have norm $1$.
    \end{enumerate}
    If $(ii)$, $(iii)$ hold but $\gamma$ is periodic and injective over its period, we call $\gamma$ a simple loop.
\end{definition}

\begin{definition}[Simple set]\label{def-5b}
    An open set $A\subset \R^2$ is simple if it is connected and its boundary is the range of a simple path or loop.
\end{definition}

\begin{definition}\label{def-5c}
    Assume $\Omega$, $W$ are transverse and $\Omega$ is simple. If $\p \Omega$ is bounded, $\XX_{\Omega, W} = 0$. If $\p \Omega$ is unbounded, let $\gamma$ be a simple path with range $\p \Omega$ such that $\Omega$ sits to the left of $\gamma$. The intersection number is defined as
\[
\XX_{\Omega, W} := \XX_+(\Omega, W) - \XX_-(\Omega, W), \qquad \XX_\pm(\Omega, W)  := \lim_{t\in \pm \infty} \1_V \circ \gamma(t).
\]
\end{definition}

\subsection[When Omega is simple]{\texorpdfstring{When $\Omega$ is simple}{When Omega is simple}}\label{subsec-5.1}
When $\Omega$ is simple, $\XX_{\Omega, W}$ is either $0$, $1$, $-1$. We prove \Cref{thm-3} for each case below:

\begin{proof}[Proof of \Cref{thm-3} when $\Omega$ is simple and $\XX_{\Omega, W} = 0$.]
By \Cref{def-5c}, when $\Omega$ is simple, $\XX_{\Omega, W}$ takes either $0$, $1$ or $-1$.

\textbf{1.} We first prove $\II_{\Omega, W}(H) = 0$ when $\XX_{\Omega, W} = 0$, i.e. $\gamma$ begins and ends both in $W$ or $W^c$. By \Cref{lemma-3e}, we can assume without loss of generality the latter case, i.e., $\p \Omega$ starts and ends in $W^c$. In this case, there exists a large enough $R$ such that $\Omega \subset W^c\cup B_R$. By \Cref{lemma-3e} and \Cref{lemma-3h},
\[
\II_{\Omega, W}(H) = \II_{\Omega, W^c}(H) = \II_{\Omega, W^c\cup B_R}(H).
\]
Hence, without loss of generality, we can assume $\Omega \subset W^c$. Then
\begin{align}
    \label{eq-5a}
\II_{\Omega, W}(H) \equiv \dim\ker(\1_{W^C} + \1_W U_\Omega \1_W) = \dim\ker(\1 + \1_W(U_\Omega - \1)\1_W).
\end{align}
Similar to \Cref{eq-2a}, we see that
\[
\1 - U_\Omega = \1 - e^{-2\pi i \1_\Omega \Pi \1_\Omega} = \1_\Omega \QQ_\Omega \1_\Omega
\]
for some operator $\QQ_\Omega$. Plug into \eqref{eq-5a} and notice $\1_W\1_\Omega = 0$ since $\Omega \subset W^c$, we obtain
\[
\II_{\Omega, W}(H) \equiv\dim\ker(\1) = 0.
\]
\end{proof}

\begin{proof}[Proof of Theorem \ref{thm-3} when $\Omega$ is simple and $\XX_{\Omega, W} = 1$.]

\textbf{1.} Recall \cite[Proposition 5]{DZ24}: for $(\Omega,W)$ simple and transverse with $\chi_{\Omega,W}=1$, there exists $\Omega_n, W_n$, see \Cref{fig:P12}, with the following properties:
\begin{enumerate}[label=($\mathcal{P}$\arabic*)]
    \item \label{P1} $\Omega \Delta \Omega_n$ and $W \Delta W_n$ are compact (where $\Delta$ denotes the symmetric difference)
    \item \label{P2} $\{ (\Omega_n, W_n), n \in \N \}$ is uniformly transverse: there exists $c \in (0,1)$ such that for all $n$,
    \begin{equation}\label{eq:2q}
        \forall |x| \geq c^{-1}, \qquad \Psi_{\Omega_n,W_n}(x) \geq |x|^c.
    \end{equation}
    \item \label{P3} $\Omega_n \cap \Bb_{8n}(0) = \{x_2 > 0\} \cap \Bb_{8n}(0)$ and $W_n \cap \Bb_{8n}(0) = \{x_1 > 0\} \cap \Bb_{8n}(0)$.
\end{enumerate}
In particular, by \ref{P1} and using the invariance of $\II$ under compact perturbations -- \Cref{lemma-3h} -- we have:
\begin{equation}\label{eq:2p}
    \II_{\Omega,W}(H) \equiv \II_{\Omega_n,W_n}(H).
\end{equation}
\begin{figure}[t]
  \centering
  \safeincludegraphics[width=0.75\textwidth]{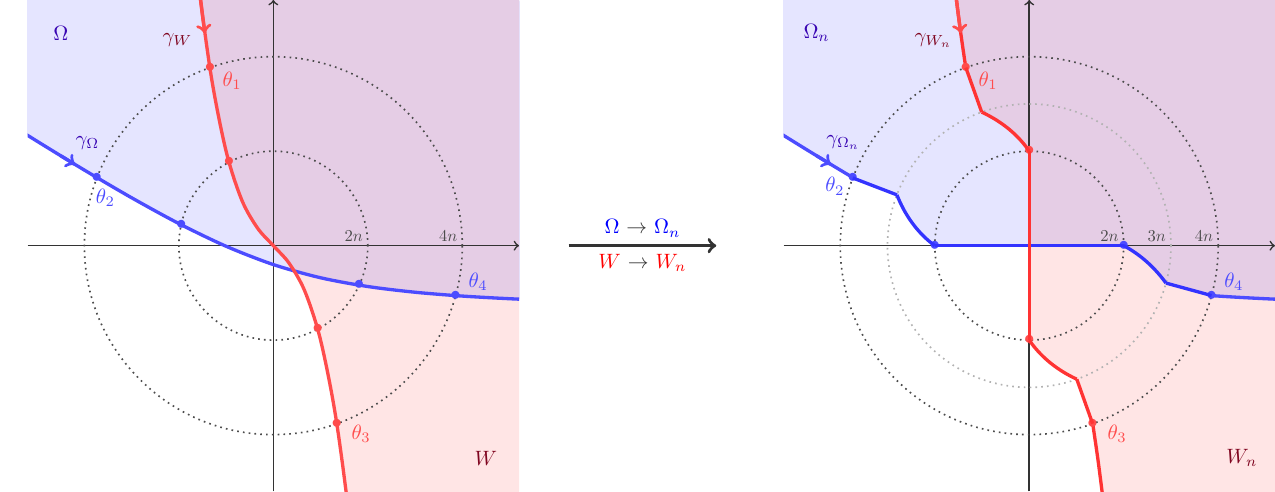}
  \caption{Deforming $\Omega, W$ to $\Omega_n, W_n$. Blue (red) region represents $\Omega$ ($W$) and $\Omega_n$ ($W_n$).}
  \label{fig:P12}
\end{figure}

\textbf{2.} Meanwhile we claim, when $n$ is large enough, $\II_{\Omega_n, W_n}(H_\pm) = \II_{\Hh_2, \Hh_1}(H_\pm)$. Recall by \eqref{eqdef-3b}, $\II_{\Omega, W}(H) = \II(U_\Omega, \1_W) \equiv \dim\ker(\1 - \AAA_{\Omega, W})$ where $\AAA_{\Omega, W}= U_\Omega^*[U_\Omega, \1_W]$. By \Cref{rmk-3a}, $\AAA_{\Omega, W}$ decay polynomially w.r.t. $\p \Omega$ and $\p W$. This implies, by \Cref{def:trans sets}, for any $N$, there is $C_N$ such that
\begin{align}
    \big| A_{\Omega_n,W_n}(x,y) \big| & \leq C_N \Psi_{\Omega_n, W_n}(x)^{-N}\Psi_{\Omega_n, W_n}(y)^{-N},
    \\
    \big| A_{\Hh_2,\Hh_1}(x,y) \big| & \leq C_N|x|^{-N}|y|^{-N},
\end{align}
where we used that $\Psi_{\Hh_1,\Hh_2}(x) = |x|$.
In particular, when $n \geq c^{-1}$, where $c \in (0,1)$ satisfies \eqref{eq:2q}, and $x,y \notin \Bb_n(0)$, by triangle inequality, we have
\begin{equation}\label{eq:2n}
    \big| A_{\Omega_n,W_n}(x,y) - A_{\Hh_2,\Hh_1}(x,y) \big| \leq C_N|x|^{-N}|y|^{-N} + C_N |x|^{-cN}|y|^{-cN} \leq 2C_N |x|^{-cN} |y|^{-cN}.
\end{equation}

\textbf{3.} On the other hand, as a consequence of \ref{P3}, we have
\begin{equation}\label{eq:2o}
   \forall x,y \in \Bb_n(0), \qquad \big| A_{\Omega_n,W_n}(x,y) - A_{\Hh_2,\Hh_1}(x,y) \big| \leq \mu^{-1}e^{-\mu n}.
\end{equation}
Indeed, given two sets $X, Y\subset \Z^2$, consider
\begin{align*}
A_{X,Y}-A_{X\cap B_n(0),Y\cap B_n(0)}
&=(U_X-U_{X\cap B_n(0)})^*[U_X,\1_Y]\\
&\quad+U_{X\cap B_n(0)}^*[U_X-U_{X\cap B_n(0)},\1_Y]\\
&\quad+U_{X\cap B_n(0)}^*[U_{X\cap B_n(0)}^*,\1_{Y\setminus B_n(0)}].
\end{align*}

\textbf{4.} Combining \eqref{eq:2n} and \eqref{eq:2o} imply that for $n \geq c^{-1}$:
\begin{align}
    \big\| A_{\Omega_n,W_n} - A_{\Hh_2,\Hh_1} \big\|_\Tr & \leq \sum_{x,y \in \Z^2} \big| A_{\Omega_n,W_n}(x,y) - A_{\Hh_2,\Hh_1}(x,y) \big|
    \\ &
    \leq \sum_{x,y \in \Bb_n(0)} \mu^{-1}e^{-\mu n} + \sum_{\substack{ x \notin \Bb_n(0) \\ \text{or } y \notin \Bb_n(0)}} e^{-\mu|x|^c-\mu|y|^c}
    \\ &
    \leq \mu^{-1} n^2 e^{-\mu n} + \mu^{-1} e^{-\mu n^c} \sum_{x \in \Z^2} e^{-\mu |x|^c}.
\end{align}
Because the above series is finite, we deduce that $A_{\Omega_n,W_n}$ converges in trace-class norm to $A_{\Hh_2,\Hh_1}$ as $n \rightarrow \infty$. In particular, we have for $n$ sufficiently large:
\begin{equation}
    \big\| A_{\Omega_n,W_n} - A_{\Hh_2,\Hh_1} \big\|_\Tr \leq \dfrac{2^{-10}}{(\| A_{\Hh_2,\Hh_1} \|_\HS^2+1)^2}.
\end{equation}
From \Cref{prop-2b}, we deduce that
$\II_{\Omega_n,W_n}(H_\pm) = \II_{\Hh_2,\Hh_1}(H_\pm) = \II(H_\pm)$. By \eqref{eq:2p}, we conclude that $\II_{\Omega,W}(H_\pm) = \II(H_\pm)$. This completes the proof of the case $\XX_{\Omega,W} = 1$.

\end{proof}

\begin{proof}[Proof of Theorem \ref{thm-3} when $\Omega$ is simple and $\XX_{\Omega, W} = -1$.]
    When $\XX_{\Omega, W}=-1$, we need to show $\II_{\Omega, W}(H_\pm) \equiv -\II(H_\pm) \equiv \II(H_\pm)$. Since passing to $\Omega^c$ reverses the orientation, $\XX_{\Omega^c,W}=1$. By the proof above, $\II_{\Omega^c, W}(H_\pm) = \II(H_\pm)$. By \Cref{lemma-3e}, $\II_{\Omega, W} = \II_{\Omega^c, W}$. Hence
    \[
    \II_{\Omega, W}(H_\pm) = \II_{\Omega^c, W}(H_\pm) = \II(H_\pm).    \]
\end{proof}

\subsection{Extension to general transversal sets}\label{subsec-5.2}
The extension of \Cref{thm-3} from simple transverse sets to general transverse sets follows the argument of \cite[Sec. 7]{DZ24}. The required inputs are the complement identity of \Cref{lemma-3e}, compact-perturbation invariance of \Cref{lemma-3h}, and additivity in $\Omega$ and $W$ from \Cref{lemma-3f,lemma-3g}.

We recall the outline. We say
\begin{definition}\label{def:4} An open subset $A$ of $\R^2$ is good if:
\begin{itemize}
    \item[(a)] The set $\Int A^c$ has boundary $\p A$;
    \item[(b)] The connected components $\{ \Gamma_k : k \in \N \}$ of $\p A$ are the ranges of simple paths or loops and satisfy $\inf\limits_{j \neq k} d(\Gamma_j, \Gamma_k) > 0$.
    \item[(c)] $\p A$ does not intersect $\Z^2$.
\end{itemize}
\end{definition}
These are sets with not-too-pathological boundaries, see \cite[Sec. 7.1]{DZ24} for basic properties of such sets. It turns out one can extend naturally the definition of intersection number $\XX_{\Omega, W}$ from transversal simple sets to transversal good sets, see \cite[Sec. 7.2]{DZ24}, which, roughly speaking, counts the signed number of times oriented boundary of $\Omega$ enters $W$.

With this definition, we can extend \Cref{thm-3} step by step, following the proof in \cite[Sec. 7.5]{DZ24}, from transverse simple sets to transverse good sets:
\begin{itemize}
    \item[A.] $U$ and $V$ are simple.
    \item[B.] $U$ is simple and $V$ is good and connected.
    \item[C.] $U$ is simple and $V$ is good.
    \item[D.] $U$ is good and connected and $V$ is good.
    \item[E.] And finally, $U$ and $V$ are good.
\end{itemize}

Finally, given any transverse sets $\Omega$, $W\subset \Z^2$, we proved in \cite[Lemma 7.1]{DZ24} that there exist transverse good sets $\overline{\Omega}$, $\overline{W}\subset \R^2$ such that $\overline{\Omega} \cap \Z^2 = \Omega$ and $\overline{W} \cap \Z^2 = W$. As a result,
\[
\II_{\Omega, W}(H_\pm) = \II_{\overline{\Omega}, \overline{W}}(H_\pm) = \XX_{\Omega, W}\II(H_\pm).
\]
This allows us to extend \Cref{thm-3} from transverse good sets to general transversal sets.

This completes the proof of \Cref{thm-3}.

\bibliographystyle{amsalpha}
\bibliography{references}

\end{document}